\documentclass[enabledeprecatedfontcommands,a4paper]{scrartcl}
\usepackage{amssymb}
\usepackage{amsmath}
\usepackage{amsthm}
\usepackage{latexsym}
\usepackage{enumerate}
\usepackage{graphicx,color}
\usepackage{hyperref}
\usepackage{actuar}
\usepackage{dsfont}
\usepackage{mathtools}

\usepackage[german,english]{babel}

\usepackage{xcolor}

\numberwithin{equation}{section}

\numberwithin{figure}{section}

\theoremstyle{plain}
\newtheorem{theorem}{Theorem}[section]
\newtheorem{proposition}[theorem]{Proposition}
\newtheorem{corollary}[theorem]{Corollary}

\theoremstyle{definition}
\newtheorem{definition}[theorem]{Definition}

\newtheorem{example}[theorem]{Example}

\renewcommand{\d}{\mathrm{d}}
\newcommand{\E}{\mathbb{E}}
\newcommand{\Pro}{\mathbb{P}}

\allowdisplaybreaks

\title{Two-dimensional forward and backward transition rates}

\author{Theis Bathke\footnote{Carl von Ossietzky Universit\"{a}t Oldenburg, Institut f\"{u}r Mathematik, 26111 Oldenburg, Germany, theis.bathke@uni-oldenburg.de}, Marcus C.~Christiansen\footnote{Carl von Ossietzky Universit\"{a}t Oldenburg, Institut f\"{u}r Mathematik, 26111 Oldenburg, Germany, marcus.christiansen@uni-oldenburg.de}}

\date{\today
}

\begin{document}

\maketitle

\begin{abstract}
Forward transition rates were originally introduced with the aim to evaluate  life insurance liabilities  market-consistently. While this idea turned out to have its limitations, recent literature repurposes forward transition rates as a tool for avoiding Markov assumptions in the calculation of  life insurance reserves. While life insurance reserves are some form of conditional first-order moments, the calculation of conditional second-order moments needs an extension of the forward transition rate concept from one dimension to two dimensions.  Two-dimensional forward transition rates are also needed for the calculation of path-dependent life insurance cash-flows as they occur upon  contract modifications. Forward transition rates are designed for doing prospective calculations, and by a time-symmetric definition of so-called backward transition rates one can do retrospective calculations.
\end{abstract}


\section{Introduction}

Forward transition rates describe the expected number of future transitions conditional on the currently available information. If the current information is incomplete,  then backward transition rates serve as a proxy for the number of past transitions. However, forward and backward transition rates do not describe the inter-temporal dependency structure of two successive jump events, so they are unsuitable for the calculation of second-order moments or for dealing with path-dependent insurance cash-flows as they occur upon contract modifications. For this reason, this paper introduces two-dimensional forward and backward transition rates and explains their use in life insurance calculations.

The research on forward transition rates originally emerged from the desire to calculate market values for life insurance liabilities. 
Miltersen and Persson \cite{miltersen2005mortality}
suggested to define mortality rates implicitly in such a way that the classical actuarial formulas reproduce market values  instead of real-world expectations. They denoted these implicit mortality rates as forward mortality rates, inspired by the concept of forward interest rates from financial mathematics. 
Norberg \cite{norberg2010forward} developed generalizations of the forward concept to multi-state life insurance but observed that the implicit definitions are lacking uniqueness. 
 Christiansen and Niemeyer  \cite{christiansen2015forward} observed that the implicit definitions often do not have any solution at all. So it turned out that the implicit  concept has serious limitations. 
 Buchard, Furrer and Steffensen \cite{buchardt2019forward} 
    suggested to overcome these limitations by introducing additional artificial states, but their concept is restricted to insurance contracts that have sojourn payments only and no transition payments.
 Buchardt \cite{buchardt2017kolmogorov} discarded  the implicit concept and presented an explicit definition that works for a doubly stochastic Markov framework. Buchardt actually drifts away from market-consistent valuation but without clearly mentioning.
 %
 Christiansen  \cite{christiansen2021calculation}
 discovered that Buchardt's approach can be extended to a general recipe for calculating real-world expectations in fully non-Markovian life insurance frameworks. Christiansen and Furrer \cite{christiansen2022extension} 
 enriched this non-Markov concept with a change of measure technique that makes it possible to deal also with path-dependent insurance payments. To put it into a nutshell, we find two divergent concepts of forward transition rates in the actuarial literature.
  First, there is the implicit concept which aims to  reproduce market values, first suggested by  Miltersen and Persson  \cite{miltersen2005mortality}  and taken on by numerous further authors. Second, there are the explicit definitions of  Buchardt \cite{buchardt2017kolmogorov},  Christiansen  \cite{christiansen2021calculation}, and Christiansen and Furrer \cite{christiansen2022extension}  
   that put the focus back on real-world expectations and use forward transition rates as a tool to cope with complex inter-temporal dependency structures in life insurance models. This paper follows the second line of thought.

The change of measure technique in Christiansen and Furrer \cite{christiansen2022extension}
  transfers the complexity of path-dependent insurance cash-flows to auxiliary probabilistic models. This way the forward transition rates may stay one-dimensional, but each insurance cash-flow needs another probabilistic model. Our aim is to have one probabilistic model only, and we achieve that by expanding the concept of forward transition rates to two dimensions. The price for having just one probabilistic model is an increased numerical effort that comes  with the extra dimension.   So the approach of Christiansen and Furrer \cite{christiansen2022extension}  
   is beneficial when many scenarios are calculated for one and the same cash-flow, whereas the results of this paper are favourable when many different cash-flows ought to be calculated. 

One-dimensional forward transition rates are designed as a tool for calculating expectations. Our two-dimensional concept can be used as a tool for calculating  second-order moments. The calculation of second-order moments in classical Markov models was first outlined in Hoem \cite{hoem1969markov}.
Helwich \cite{helwich2008durational} presents general formulas for the calculation of variances in semi-Markov models.
Calculation formulas for second-order moments in fully non-Markovian models do not exist yet in the literature. By introducing  two-dimensional forward and backward transition rates we help  to close that gap.

Markov modelling has a long tradition in life insurance. The classical concept to model the random pattern of states of the insured as a Markov process has been extended to  semi-Markov modelling and further sophisticated Markov structures by numerous authors. The problem with any kind of  Markov assumptions is that they come with model risk, while at the same time there is a lack of tools for quantifying the model error. This motivates the search for non-Markovian calculation techniques, and  this paper is a major step into that direction.

The paper is structured as follows:
Section \ref{Chap:2} introduces the fundamental life insurance modelling framework.
Section \ref{Chap:3 forward rates} introduces the definition of two-dimensional forward and backward transitions rates and develops a corresponding integral equation that generalizes Kolmogorov's forward equation.  For the latter integral equation  section \ref{chap:4 solutions} verifies the uniqueness of  solutions.
Section \ref{Section:ExpectationFormulas} then turns to the main purpose of two-dimensional forward and backward transitions rates, namely the calculation of certain conditional moments.
In Section \ref{Chap:5 prospective reserve} we explain the calculation of conditional variances. 
Section  \ref{Chap:7 retrospective calculation} illustrates the calculation of path-dependent cash-flows. Section \ref{SectionConclusion} concludes and gives an outlook on open research questions.

\section{Life insurance modelling framework}\label{Chap:2}

Let $(\Omega,\mathcal{A},\mathbb{P})$ be a probability space with a filtration $\mathbb{F}=(\mathcal{F}_t)_{t \geq 0}$. We consider an individual life insurance contract and describe the status of the individual insured by an adapted c\`{a}dl\`{a}g  jump process
\begin{align*}
Z=(Z(t))_{t\geq 0}
\end{align*}
on a finite state space $\mathcal{Z}$. Additionally, we set $Z_{0-}:=Z_0$. Throughout this paper we assume that we are currently at time $s \geq 0$. So the time interval $[0,s]$ represents to the past and the present, and the time interval $(s,\infty)$ represents the future. The real number $s \geq 0$ is arbitrary but fixed.
 Based on $Z$ we define state indicator processes $(I_i)_{i\in\mathcal{Z}}$ by
\begin{align*}
&I_i(t):=\mathds{1}_{\{Z(t)=i\}},\quad t \geq 0,
\end{align*}
  and transition counting processes $(N_{ij})_{i,j\in\mathcal{Z}:i\neq j}$ by
\begin{align*}
&N_{ij}(t):=\#\{u\in(0,t]:Z(u-)=i,Z(u)=j\},\quad t \geq 0.
\end{align*}
We generally assume that
\begin{align}
\mathbb{E}[N_{ij}(t)]<\infty,\quad t\geq 0\;,i,j\in\mathcal{Z},\; i\neq j,\label{eq: finite expected value of N}
\end{align}
which in particular implies that $Z$ has almost surely no explosions.
Let
\begin{align} \label{eq: DiagSprung} \begin{split}
N_{ii}(t)&:=-\sum_{j\in \mathcal{Z}}(N_{ij}(t)-N_{ij}(s)),\quad  t>s,  \\
N_{ii}(t)&:=-\sum_{j\in \mathcal{Z}}(N_{ji}(t)-N_{ji}(s)),\quad t\leq s.
\end{split}
\end{align}
This construction is made so that $N_{ii}(t)$ satisfies 
\begin{align*}
	\d N_{ii}=-\sum_{j\in\mathcal{Z}}\d N_{ij}, \quad t>s,\\
	\d N_{ii}=-\sum_{j\in\mathcal{Z}}\d N_{ji}, \quad  t\leq s,
\end{align*}
	 and is c\`{a}dl\`{a}g everywhere.
Definition \eqref{eq: DiagSprung} and many further definitions following below involve a case differentiation between $t>s$ and $t \leq s$. As time $s$ is fixed, we omit it in the notation, but one should keep in mind the dependence of many of our definitions on the parameter $s$.
The following equations show a useful direct link between the processes $(N_{ij})_{i,j\in\mathcal{Z}}$ and $(I_i)_{i\in\mathcal{Z}}$:
\begin{align}
I_i(t)=I_i(s)+\sum_{j\in \mathcal{Z}} \;\int\displaylimits_{(s,t]}  N_{ji}(\d u),\quad  t\geq s, i\in\mathcal{Z}\label{I_represent_by_N}.
\end{align} and
\begin{align}\label{RelationBetweenIandN}
	I_i(t)=I_i(s)+\sum_{j\in \mathcal{Z}} \;\int\displaylimits_{(t,s]}  N_{ij}(\d u),\quad s\geq t
	, i\in\mathcal{Z}.
\end{align}
The latter integrals and all following integrals in this paper are meant as path-wise Lebesgue-Stieltjes integrals.

The sigma-algebra $\mathcal{F}_s$ represents the available information at time $s$. In insurance practice, the insurer often uses a reduced information set $\mathcal{G}_s$ for actuarial evaluations.  In this paper we generally assume that
\begin{align}\label{mathcalGs}
\sigma(Z(s))\subseteq \mathcal{G}_s\subseteq \mathcal{F}_s.
\end{align}
The special case $\mathcal{G}_s= \sigma(Z(s))$ is known as the as-if-Markov model, since it uses information of Markov-type only.
The choice of $\mathcal{G}_s$ can be influenced by many factors. Some of these are listed here: \begin{itemize}
\item By cutting down the information used, one can  reduce the complexity of numerical calculations.
\item A lack of data for statistical inference may make it necessary to simplify the information model.
\item For some actuarial tasks it is sufficient to study mean portfolio values only, and then it is convenient to minimize the individual information.
\item Anti-discrimination regulation may be a limiting factor for the use of information, such as unisex calculation. 
\item Data privacy regulation can restrict the information that the insurer actually gathers and stores. For example, the General Data Protection Regulation of the European Union gives individuals considerable  control and rights over their personal data.
\end{itemize}
Let $B$ be the insurance cash-flow of the individual life insurance contract, here assumed to be an adapted c\`{a}dl\`{a}g process with paths of finite variation. We assume that the insurance contract has a maximum contract horizon of $T$, which means that
\begin{align*}
 B(\d t) = 0 , \quad t  > T.
\end{align*}
Throughout this paper we assume that
\begin{align*}
 0 \leq  s \leq T < \infty .
\end{align*}
Let $\kappa$ be a c\`{a}dl\`{a}g function that describes the value development of a savings account. We assume that $\kappa$ is strictly positive so that the corresponding discounting function $1/\kappa$ exists. 
The random variable 
\begin{align}
Y^+:=\;\int\displaylimits_{(s,T]}\frac{\kappa(s)}{\kappa(u)} B(\d u)\label{def:DefOfY}
\end{align}
describes the discounted accumulated future cash-flow seen from time $s$, 
and the random variable
\begin{align}
Y^-:=\;\int\displaylimits_{[0,s]}\frac{\kappa(s)}{\kappa(u)}  B(\d u)\label{def:DefOfY-}
\end{align}
is the compounded accumulated past  cash-flow. In order to ensure integrability, we generally assume that  $1/\kappa(\cdot)$ is bounded on $[0,T]$.

Classical Markov modelling focusses on benefit payment functions that may depend on current states or current jumps of the insured but not on  past random events. This means that $B$ is supposed to have a so-called one-dimensional canonical representation.
\begin{definition}
A stochastic process $A$ is said to have a \emph{one-dimensional canonical representation} if  there exist real-valued functions $(A_{i})_{i}$ on  $[0,\infty)$ that generate  finite signed measures $A_i(\d u)$ and there exist measurable and bounded real functions $(a_{ij})_{ij:i\neq j}$ such that
\begin{align}
 A(t)=\sum_{i \in \mathcal{Z}} \;\int\displaylimits_{[0,t]} I_i(u^-) A_i(\d u)+\sum_{i,j\mathcal{Z} \atop i\neq j}\;\int\displaylimits_{[0,t]} a_{ij}(u) N_{ij}(\d u),\quad t\geq 0.\label{def:DefOfB}
\end{align}
\end{definition}
Insurance cash-flows that involve  contract modifications, such as a free-policy option, can not be brought into the form \eqref{def:DefOfB}. For this reason we need to allow for more complex structures.
\begin{definition}
A stochastic process $A$ is said to have a \emph{two-dimensional  canonical representation} if  there exist  real-valued functions $(A_{i})_{i}$ on  $[0,\infty)$ that generate  finite signed measures $A_i(\d u_1)$, real-valued functions  $(A_{ij})_{ij}$ on $[0,\infty)^2$  that generate finite signed measures $A_{ij}(\d u_1,\d u_2)$, and measurable and bounded real-valued functions $(a_{ikl})_{ikl:k\neq l}$, $(a_{ijkl})_{ijkl:i\neq j, k\neq l}$ on  $[0,\infty)^2$ such that
\begin{align}\begin{split}
A(t)&=\sum_{i,j \in \mathcal{Z}}\; \;\int\displaylimits_{[0,t]^2} I_i(u_1^-)I_j(u_2^-) A_{ij}(\d u_1, \d u_2)\\
&\quad + \sum_{\substack{i,k,l \in\mathcal{Z} \\k\neq l}} \;\int\displaylimits_{[0,t]^2} I_i(u_1^-)a_{ikl}(u_1,u_2) A_{i}(\d u_1)   N_{kl}(\d u_2)\\
&\quad +\sum_{\substack{i,j,k,l \in \mathcal{Z}\\i \neq j, k\neq l}}\;\int\displaylimits_{[0,t]^2]} a_{ijkl}(u_1,u_2)   N_{ij}(\d u_1) N_{kl}(\d u_2),\quad t_1,t_2 \geq 0.\label{def:DefOfB2}
\end{split}
\end{align}
\end{definition}
\begin{example}\label{ExampleSquaredCashFlow}
Suppose that $B$ is an insurance cash flow that has a one-dimensional canonical representation with respect to suitable functions $(B_i)_i$ and $(b_{ij})_{ij:i \neq j}$.  Then the squared process $B^2$ has a two-dimensional canonical representation:
\begin{align}\begin{split}
  B^2(t)   &=\sum_{i,j \in \mathcal{Z}} \;\int\displaylimits_{[0,t]^2} I_i(u_1^-)I_j(u_2^-) B_{i}(\d u_1) B_j(\d u_2)\\
&\quad + \sum_{\substack{i,k,l \in \mathcal{Z}\\ k\neq l}} \;\int\displaylimits_{[0,t]^2} 2I_i(u_1^-) b_{kl}(u_2) B_{i}(\d u_1)   N_{kl}(\d u_2)\\
&\quad +\sum_{\substack{i,j,k,l \in \mathcal{Z}\\i \neq j, k\neq l}}\;\int\displaylimits_{[0,t]^2} b_{ij}(u_1) b_{kl}(u_2)   N_{ij}(\d u_1) N_{kl}(\d u_2),\quad t\geq 0. \label{def:RepOfBsquared}
\end{split}\end{align}
In order to see this, use integration by parts according to Proposition \ref{prop: partial integration} and Fubini's theorem.
\end{example}
\begin{example}\label{ExampleFPOCashFlow}
Suppose that $\mathcal{Z}= \mathcal{S}  \times \{0,1\}$, where the elements of $\mathcal{S}$ describe the health status of the individual insured and $\{0,1\}$ indicates whether the policyholder has exercised a free-policy option. We assume here that the free-policy option can be exercised only once and that the policy cannot move back to a premium-paying status. Let $\tau$ be the random time where the free-policy option is actually exercised, i.e.~$\tau$ gives the time where $Z$ moves from the set $\mathcal{S}_0:=\mathcal{S} \times \{0\}$ to the set $\mathcal{S}_1:=\mathcal{S}\times \{1\}$.  At time $\tau$ the insurance payment scheme is rescaled by a factor $\rho(\tau,Z_{\tau^-},Z_\tau)$ in order to maintain actuarial equivalence, cf. \cite{christiansen2022extension}.  So, by writing $C$ for the insurance payment scheme, the insurance cash-flow $ B$ equals
\begin{align*}
 B( t) = \;\int\displaylimits_{[0,t]} \rho(\tau,Z_{\tau^-},Z_\tau)^{\mathds{1}_{\{\tau \leq u\}}}  C(\d u).
\end{align*}
We assume  that $C$ has a finite horizon of $T$. For the sake of simplicity, let
\begin{align}\label{CiZeroAtTau}
\mathds{1}_{\{\tau=u\}} C(\d u)=0,\quad\forall i\in\mathcal{Z},
\end{align}
almost surely, so that we almost surely have no lump sum payments at time $\tau$.
Suppose that the payment scheme $C$  has a one-dimensional canonical representation  with respect to suitable functions $(C_i)_i$ and $(c_{ij})_{ij:i \neq j}$. The cash-flow $ B$ can be decomposed to
\begin{align}\begin{split}
	 B( t) =& \;\int\displaylimits_{[0,t]}\mathds{1}_{\{u<\tau\}}   C(\d u)+\;\int\displaylimits_{[0,t]}\mathds{1}_{\{u\geq\tau\}} \rho(\tau,Z_{\tau^-},Z_\tau) C(\d u).\end{split}\label{eq:FirstDecompositionFreePolicyOption}
\end{align}
We will now analyse both integrals separately. Beginning with the first one, we have
\begin{align*}
\int\displaylimits_{[0,t]}\mathds{1}_{\{u<\tau\}}  C(\d u)&=\sum_{i\in \mathcal{S}_0}\;\int\displaylimits_{[0,t]}\mathds{1}_{\{u<\tau\}} I_i(u^-) C_i(\d u)+\sum_{\substack{i,j\in \mathcal{S}_0\\i\neq j}}\;\int\displaylimits_{[0,t]}\mathds{1}_{\{u<\tau\}}c_{ij}(u)  N_{ij}(\d u)
\end{align*}
since $\mathds{1}_{\{u<\tau\}}I_i(u^-)=0$ and  $\mathds{1}_{\{u<\tau\}}N_{ij}(\d u)=0$ for all $(i,j)\not\in \mathcal{S}_0^2$.
Because of \eqref{CiZeroAtTau} and the fact that  $\mathds{1}_{\{u\leq \tau\}}I_i(u^-)=I_i(u^-)$ and  $\mathds{1}_{\{u<\tau\}}N_{ij}(\d u)=N_{ij}(\d u)$ for $i,j\in \mathcal{S}_0$, we furthermore get
\begin{align*}
\int\displaylimits_{[0,t]}\mathds{1}_{\{u<\tau\}}  C(\d u)=\sum_{i\in \mathcal{S}_0}\;\int\displaylimits_{[0,t]} I_i(u^-) C_i(\d u)+\sum_{\substack{i,j\in \mathcal{S}_0\\i\neq j}}\;\int\displaylimits_{[0,t]}c_{ij}(u)  N_{ij}(\d u)
\end{align*}
almost surely.
So the insurance cash-flow prior to $\tau$ almost surely has a one-dimensional canonical representation. Since $\tau$ is the unique jump time of the counting process $ \sum_{k \in \mathcal{S}_0}\sum_{l \in \mathcal{S}_1} N_{kl}$, by using assumption \eqref{CiZeroAtTau} the second integral in \eqref{eq:FirstDecompositionFreePolicyOption} can be almost surely transformed to
\begin{align*}
&\int\displaylimits_{[0,t]}\mathds{1}_{\{u\geq\tau\}} \rho(\tau,Z_{\tau^-},Z_\tau) C(\d u) \\
&=  \sum_{k \in \mathcal{S}_0}\sum_{l \in \mathcal{S}_1} \;\int\displaylimits_{[0,t]^2} {\mathds{1}_{\{u_1 \geq u_2\}}}\rho(u_2,k,l)  C(\d u_1)  N_{kl}(\d u_2)\\
&=  \sum_{k \in \mathcal{S}_0}\sum_{i,l \in \mathcal{S}_1} \;\int\displaylimits_{[0,t]^2} I_i(u_1^-)\rho(u_2,k,l) C_{i}(\d u_1)   N_{kl}(\d u_2)\\
&\quad +\sum_{k \in \mathcal{S}_0}\sum_{\substack{i,j,l\in \mathcal{S}_1\\i \neq j}}\;\int\displaylimits_{[0,t]^2} \rho(u_2,k,l) c_{ij}(u_1)    N_{ij}(\d u_1)  N_{kl}(\d u_2).
\end{align*}
So, at and after time $\tau$ the insurance cash-flow almost surely has a two-dimensional canonical representation.  All in all,  we obtain for the insurance cash-flow $B$ the almost sure representation
\begin{align}\begin{split}
	 B(t)&=\sum_{i\in \mathcal{S}_0}\;\int\displaylimits_{[0,t]} I_i(u^-) C_i(\d u)+\sum_{\substack{k,l\in \mathcal{S}_0\\k\neq l}}\;\int\displaylimits_{[0,t]}c_{kl}(u)  N_{kl}(\d u)\\
	 &\quad +  \sum_{k \in \mathcal{S}_0}\sum_{i,l \in \mathcal{S}_1} \;\int\displaylimits_{[0,t]^2} I_i(u_1^-)\rho(u_2,k,l) C_{i}(\d u_1)   N_{kl}(\d u_2)\\
	 &\quad +\sum_{k \in \mathcal{S}_0}\sum_{\substack{i,j,l\in \mathcal{S}_1\\i \neq j}}\;\int\displaylimits_{[0,t]^2} \rho(u_2,k,l) c_{ij}(u_1)    N_{ij}(\d u_1)  N_{kl}(\d u_2), \quad t \geq 0.\label{eq: free-policy-option-cash-flow}
	 \end{split}
	\end{align}
\end{example}

\section{Two-dimensional transition rates}\label{Chap:3 forward rates}

This section generalizes the forward and backward transition rates of \cite{christiansen2021calculation} from one dimension  to two dimensions.  We still suppose that we are currently at time $s$ and have the information $\mathcal{G}_s$ available.  As the parameter $s$ is  fixed we omit in the notation.

Let $P_i=(P_i(t))_{t\geq 0}$ and $Q_{ij}=(Q_{ij}(t))_{t \geq 0}$ for $ i,j \in \mathcal{Z}$ be the almost surely unique c\`{a}dl\`{a}g  processes that  satisfy
\begin{align*}
  P_i(t)&=\E[ I_i(t) | \mathcal{G}_s],\\
  Q_{ij}(t) &= \E[ N_{ij}(t)| \mathcal{G}_s].
\end{align*}
 As we already mentioned after  definition \eqref{eq: DiagSprung}, the process $N_{ii}$ is  c\`{a}dl\`{a}g.

Let $P_{ik}= (P_{ik}(t_1,t_2))_{t_1,t_2 \geq 0}$ and $Q_{ijkl}=(Q_{ijkl}(t_1,t_2))_{t_1,t_2 \geq 0}$ for $i,j,k,l \in\mathcal{Z}$ be the almost surely unique random surfaces that are c\`{a}dl\`{a}g in each variable and satisfy
\begin{align*}
P_{ik}(t_1,t_2)&=\mathbb{E}[I_i(t_1)I_k(t_2)|\mathcal{G}_s],\\
Q_{ijkl}(t_1,t_2)&=\mathbb{E}[ N_{ij}(t_1)N_{kl}(t_2)|\mathcal{G}_s].
\end{align*}
%
In order to find suitable modifications of the conditional expectations so that the processes  have the postulated c\`{a}dl\`{a}g path properties, one can calculate the conditional expectations on the basis of a fixed regular conditional distribution $\Pro( \, \cdot \, | \mathcal{G}_s)$. Then
the c\`{a}dl\`{a}g path properties follow directly from the c\`{a}dl\`{a}g path properties of $(I_i)_i$ and $(N_{ij})_{ij}$ and the dominated convergence theorem. The c\`{a}dl\`{a}g path properties imply that the processes $P_i, Q_{ij}$ and the surfaces $P_{ik},Q_{ijkl}$ are uniquely given by their values at rational time points only, which are countably many, so the whole processes and surfaces are almost surely unique.

In the following the superscript '$\pm $' is a short notation for
\begin{align*}
	f(u^\pm):=\begin{cases}
		f(u^-)&:\, u>s,\\
		f(u)&:\, u\leq s,
	\end{cases}
\end{align*}
for c\`{a}dl\`{a}g functions $f$, and for any time points $s,t$ we define the symbol
\begin{align*}
	((s,t]]:=\begin{cases}
		(s,t] &:\,t>s,\\
		(t,s] &:\, t\leq s.
	\end{cases}
\end{align*}
Moreover, for the sake of brief formulas, we introduce the processes
 \begin{align*}
	\tilde{P}_{ij}(u):=\begin{cases}
		P_i(u) &:\,u>s,\\
		P_j(u) &:\,u\leq s,
	\end{cases}
\end{align*}
and
\begin{align*}
	\tilde{P}_{ijkl}(u_1,u_2):=\begin{cases}
		P_{ik}(u_1,u_2) &:\,u_1,u_2>s,\\
		P_{jk}(u_1,u_2) &:\,u_1\leq s,u_2>s,\\
		P_{il}(u_1,u_2) &:\,u_1> s,u_2\leq s,\\
		P_{jl}(u_1,u_2) &:\,u_1,u_2\leq s.
	\end{cases}
\end{align*}
\begin{definition}\label{DefLambda}
For $ i, j \in \mathcal{Z}$ let the stochastic process  $(\Lambda_{ij}(t))_{t \geq 0}$ be defined by
\begin{align}\label{DefEqLambda}
   \Lambda_{ij}(t) = \;\int\displaylimits_{((s,t]]} \frac{\mathds{1}_{\{\tilde{P}_{ij}(u^\pm)>0\}}}{\tilde{P}_{ij}(u^\pm)}  Q_{ij}(\d u),
\end{align}
and for $i,j,k,l \in \mathcal{Z}$ let the random surface   $(\Lambda_{ijkl}(t_1,t_2))_{t_1, t_2 \geq 0}$
be defined by
\begin{align}\label{DGL of g and Gamma}
 \Lambda_{ijkl}(t_1,t_2) = \;\int\displaylimits_{((s,t_1]]\times ((s,t_2]]} \frac{\mathds{1}_{\{\tilde{P}_{ijkl}(u_1^\pm,u_2^\pm)> 0\}}}{\tilde{P}_{ijkl}(u_1^\pm,u_2^\pm)}Q_{ijkl}(\d u_1,\d u_2).
\end{align}
\end{definition}
In case of $t_1,t_2>s$ we denote $ \Lambda_{ij}(\d t_1 )$ as (one-dimensional) forward transition rate and $\Lambda_{ijkl}(\d t_1,\d t_2)$ as two-dimensional forward transition rate. In case of $t_1,t_2\leq s$ we call them
backward transition rates.  In order to ensure existence of \eqref{DefEqLambda} and \eqref{DGL of g and Gamma}  we generally assume that
\begin{align}\label{IntegrabilCond}\begin{split}
	&\;\int\displaylimits_{[0,t]} \frac{\mathds{1}_{\{\tilde{P}_{ij}(t^\pm)>0\}}}{\tilde{P}_{ij}(t^\pm)}  Q_{ij}(\d t)  < \infty, \quad t \geq 0,\\
	&\;\int\displaylimits_{[0,t]^2} \frac{\mathds{1}_{\{\tilde{P}_{ijkl}(t_1^\pm,t_2^\pm)> 0\}}}{\tilde{P}_{ijkl}(t_1^\pm,t_2^\pm)}Q_{ijkl}(\d t_1,\d t_2)<\infty,\quad t \geq 0,
\end{split}\end{align}
almost surely for all $i,j,k,l\in\mathcal{Z}$.
According to \cite{christiansen2021calculation} it holds that
\begin{align}\label{eq: GenKolmForwardEquation}
P_i(t)=I_i(s)+\sum_{j\in\mathcal{Z}}\;\int\displaylimits_{((s,t]]} {P}_{j}(u^\pm)\tilde{\Lambda}_{ji}(\d u), \quad t \geq 0,
\end{align}
for all $i \in \mathcal{Z}$, where 
\begin{align*}
\tilde{\Lambda}_{ij}(u):=\begin{cases}
	\Lambda_{ij}(u)&: \,u>s,\\
	\Lambda_{ji}(u)&: \, u\leq s.
\end{cases}
\end{align*}
This is a generalization of  Kolmogorov's forward equation to non-Markov models.
We now extend \eqref{eq: GenKolmForwardEquation} from one to  two dimensions.
\begin{theorem}\label{theo: two dimensional Kolmogorov forward equation}
The processes $(P_{ik})_{i,k\in\mathcal{Z}}$, $(\Lambda_{ijkl})_{i,j,k,l\in\mathcal{Z}}$,$(P_{i})_{i\in\mathcal{Z}}$ and $(\Lambda_{ij})_{i,j\in\mathcal{Z}}$ almost surely satisfy the equation
\begin{align}\label{2dimKolmForwEq}\begin{split}
P_{ik}(t_1,t_2)&=I_i(s)I_k(s)
+I_k(s)\;\int\displaylimits_{((s,t_1]]}\sum_{j\in\mathcal{Z}
}P_{j}(u^\pm) \tilde{\Lambda}_{ji}(\d u) +I_i(s)\;\int\displaylimits_{((s,t_2]]}\sum_{l\in\mathcal{Z}
}P_{l}(u^\pm) \tilde{\Lambda}_{lk}(\d u)\\
&\quad +\;\int\displaylimits_{((s,t_1]]\times ((s,t_2]]}\sum_{\substack{j,l\in\mathcal{Z}}}{P}_{jl}(u_1^\pm,u_2^\pm) \tilde{\Lambda}_{jilk}(\d u_1, \d u_2)
\end{split}\end{align}
for $t_1,t_2 \geq 0$ and $i,k \in \mathcal{Z}$, where
\begin{align*}
	\tilde{\Lambda}_{ijkl}(u_1,u_2):=\begin{cases}
		\Lambda_{ijkl}(u_1,u_2)& u_1,u_2>s\\
		\Lambda_{jikl}(u_1,u_2)& u_1\leq s,u_2>s\\
		\Lambda_{ijlk}(u_1,u_2)& u_1>s,u_2\leq s\\
		\Lambda_{jilk}(u_1,u_2)& u_1,u_2\leq s.
	\end{cases}
\end{align*}
\end{theorem}
\begin{proof}
We show the proof for the case $t_1,t_2>s$ only. For the other cases the proof is similar.
 Let $\mathbb{P}_\omega(\cdot)$ be a regular version of the conditional distribution $\mathbb{P}(\cdot|\mathcal{G}_s)$ and let $\mathbb{E}_\omega[\cdot]$ be the Lebesgue-Stieltjes integral of the argument with respect to the integrator $\mathbb{P}_\omega$. By  applying Campbell's theorem, see section \ref{theo:Campbell}, and Fubini's theorem,  we get for $(a,b],(c,d]\subset [s,\infty)$ and almost all $\omega\in\Omega$
\begin{align}
\;\int\displaylimits_{(a,b]\times(c,d]}&\mathds{1}_{\{P_{ik}(u_1^-,u_2^-)(\omega)>0\}}Q_{ijkl}(\d(u_1,u_2))(\omega)\nonumber\\
&=\mathbb{E}_\omega\left[\;\int\displaylimits_{(a,b]\times(c,d]}\mathds{1}_{\{P_{ik}(u_1^-,u_2^-)(\omega)>0\}} N_{ij}(\d u_1) N_{kl}(\d u_2)\right]\nonumber\\
&=\mathbb{E}_\omega\left[\;\int\displaylimits_{(a,b]\times(c,d]}I_i(u_1^-)I_k(u_2^-) N_{ij}(\d u_1) N_{kl}(\d u_1)\right]\nonumber\\
&=\;\int\displaylimits_{(a,b]\times(c,d]}Q_{ijkl}(\d(u_1,u_2))(\omega)\label{eq:indicator function in definition}
\end{align}
since $P_{ik}(u^-_1,u^-_2)(\omega)=0$ implies $I_i(u^-_1)I_k(u^-_2)=0$ $\mathbb{P}_\omega(\cdot)$-almost surely and  $ I_i(u^-_1) N_{ij}(\d u_1) = N_{ij}(\d u_1)$.
By applying \eqref{I_represent_by_N}, we  get for $t_1,t_2\geq s$ and $i,k\in\mathcal{Z}$
\begin{align*}
P_{ik}(t_1,t_2)&=\mathbb{E}[I_i(t_1)I_k(t_2)|\mathcal{G}_s]\\
&=\mathbb{E}\Big[I_i(s)I_k(s)\Big|\mathcal{G}_s\Big]+\mathbb{E}\Big[ I_k(s)\;\int\displaylimits_{(s,t_1]}\sum_{\substack{j\in\mathcal{Z}}} N_{ji}(\d t)\Big|\mathcal{G}_s\Big]
+\mathbb{E}\Big[ I_i(s)\;\int\displaylimits_{(s,t_2]}\sum_{l\in\mathcal{Z}}N_{lk}(\d t)\Big|\mathcal{G}_s\Big]\\
&\quad+\mathbb{E}\Big[\sum_{\substack{j,l\in\mathcal{Z}}}\;\int\displaylimits_{(s,t_1]\times (s,t_2]}    N_{ji}(\d t_1) N_{lk}(\d t_2)\Big|\mathcal{G}_s\Big].
\end{align*}
By using the assumption \eqref{mathcalGs}, the definition of  $(Q_{ijkl})_{ijkl}$, the definition of $(Q_{ij})_{i,j}$, Fubini's theorem, and Campbell's theorem, we can conclude that
\begin{align*}
P_{ik}(t_1,t_2)&=I_i(s)I_k(s)+I_k(s)\;\int\displaylimits_{(s,t_1]}\sum_{\substack{j\in\mathcal{Z}}} Q_{ji}(\d t)+I_i(s)\;\int\displaylimits_{(s,t_2]}\sum_{l\in\mathcal{Z}} Q_{lk}(\d u)\\
&\quad+\sum_{j,l\in\mathcal{Z}}\quad\;\int\displaylimits_{(s,t_1]\times (s,t_2]} Q_{jilk}(\d(u_1,u_2))
\end{align*}
almost surely. The assertion follows now from \eqref{DGL of g and Gamma} and \eqref{eq:indicator function in definition}.
\end{proof}

\section{Uniqueness of solutions of the integral equations}\label{chap:4 solutions}

Equation  \eqref{eq: GenKolmForwardEquation} is commonly used for the calculation of the state occupation probabilities $(P_i)_i$ from given one-dimensional transition rates.  Likewise, equation \eqref{2dimKolmForwEq} may be used in order to calculate  the two-dimensional state occupation probabilities $(P_{ik})_{ik}$ from the one-dimensional and two-dimensional transition rates, but it is crucial then that $(P_{ik})_{ik}$  is the only solution.
\begin{theorem}\label{Theorem O ist kontraktion}
There exists an almost surely unique solution $(P_{i})_{i}$, $(P_{ik})_{ik}$  to the stochastic integral equation system formed by  equations \eqref{eq: GenKolmForwardEquation} and \eqref{2dimKolmForwEq}.
\end{theorem}
\begin{proof}
The existence of a solution follows from Theorem \ref{theo: two dimensional Kolmogorov forward equation}, so it remains to show that the solution is almost surely unique. As the equations \eqref{eq: GenKolmForwardEquation} and \eqref{2dimKolmForwEq} are almost surely pathwise integral equations,  in the remaining proof we identify without loss of generality all stochastic processes and random surfaces with just one of their paths. These paths have to be chosen such that  \eqref{IntegrabilCond} holds.

 We show the proof for the case $t_1,t_2>s$ only. For the other cases the proof is similar.   We are going to use a fixed-point argument.
For any choice of $T\in (s,\infty)$, the set
\begin{align*}
BV_2^{|\mathcal{Z}|}:=\Big\{ f:[s,T]\times[s,T]&\rightarrow\mathbb{R}^{|\mathcal{Z}|^2}\Big|\text{ there exist} \text{ finite signed measures }\mu_{ik}\\
&\text{ with }f_{ik}(x,y)=\mu_{ik}([s,x]\times[s,y]),x,y\in [s,T]\times [s,T]\Big\}
\end{align*}
is a linear space.
The Jordan-Hahn decomposition offers for any finite signed measure a unique decomposition into a difference of two finite measures, and this decomposition can be furthermore used to decompose also any $f \in BV_2^{|\mathcal{Z}|}$ into a difference $f=f^+-f^-$ of nonnegative mappings $f^+,f^- \in BV_2^{|\mathcal{Z}|}$. Based on this unique construction of $f^+$ and $f^-$, we then define $|f|:= f^++f^-$. By equipping $BV_2^{|\mathcal{Z}|}$ with the norm
\begin{align*} \Vert (f_{ik})_{i,k\in\mathcal{Z}}\Vert:=\sum_{i,k\in\mathcal{Z}}\;\int\displaylimits_{[s,T]\times[s,T]}|f_{ik}|(\d t_1,\d t_2)+\sum_{i,k\in\mathcal{Z}}| f_{ik}(s,s)|
\end{align*}
we obtain a metric space.
On this metric space, we define an  operator $O:BV_2^{|\mathcal{Z}|}\rightarrow BV_2 ^{|\mathcal{Z}|}$ as follows:
\begin{align*}
\left(O((f_{jl})_{j,l\in\mathcal{Z}})\right)_{ik}(t_1,t_2):&=\sum_{{j,l\in\mathcal{Z}}}\quad\;\int\displaylimits_{(s,t_1]\times(s,t_2]}f_{jl}(u_1^-,u_2^-) \Lambda_{jilk}(\d u_1,\d u_2)
\end{align*}
for $t_1,t_2\in[s,T]$. We want this operator to be a contraction, but unfortunately this is not true, so we need to replace our norm for $BV_2^{|\mathcal{Z}|}$ by another equivalent norm.
Let
\begin{align*}
\nu(t_1,t_2):=4\sum_{\substack{i,j,k,l\in\mathcal{Z}\\i\neq j\\ k\neq l}}\Lambda_{ijkl}(t_1,t_2), \quad t_1,t_2 \in [s,T].
\end{align*}
Because of assumption \eqref{IntegrabilCond}, each $\Lambda_{ijkl}$, $i\neq j$, $k\neq l$, has an associated   measure ${\mu}_{\Lambda_{ijkl}}$ that is  almost surely finite on $[s,T]^2$. Hence, there exists a finite  measure $\mu$ that satisfies
\begin{align}
\nu(t_1,t_2)=\mu([s,t_1]\times[s,t_2]),\quad t_1,t_2\in[s,T]\label{connection beetween mu and nu}.
\end{align}
We moreover define
\begin{align*}
\nu_1(u):&=\mu([s,u]\times[s,T]),\\
\nu_2(u):&=\mu([s,T]\times[s,u]), \quad u \in [s,T],
\end{align*}
which are c\`{a}dl\`{a}g by construction. For each $K \in (0,\infty)$  the mapping $\Vert\cdot\Vert_K$ defined by
\begin{align*}
\Vert (f_{ij})_{i,j\in\mathcal{Z}}\Vert_K:=\sum_{i,j\in\mathcal{Z}}	\;\int\displaylimits_{[s,T]\times[s,T]}e^{-K(\nu_1(u_1)+\nu_2(u_2))}|f_{ij}|(\d u_1,\d u_2)+\sum_{i,k\in\mathcal{Z}}| f_{ik}(s,s)|
\end{align*}
is a norm on $BV_2^{|\mathcal{Z}|}$ that is equivalent to the norm $\| \cdot \|$. This construction of $\Vert\cdot\Vert_K$ is inspired by \cite{christiansen2010biometric} but is here extended to the two-dimensional case.
For any  $f \in BV_2^{|\mathcal{Z}|}$,  $(a,b]\times (c,d]\subset[s,T]\times[s,T]$ and $k,i\in\mathcal{Z}$, the definitions of operator $O$ and the two-dimensional transition rates $(\Lambda_{ijkl})_{i,j,k,l\in\mathcal{Z}}$ in conjunction with the  triangle inequality yield that
\begin{align*}
\int\displaylimits_{(a,b]\times(c,d]}|O(f)_{ik}|(\d u_1,\d u_2)
\leq \sum_{\substack{l:l\neq k\\j:j\neq i}}\bigg(&
\;\int\displaylimits_{(a,b]\times(c,d]}|f_{jl}(u_1^-,u_2^-)|\Lambda_{jilk}(\d u_1,\d u_2)\\
&\quad+\;\int\displaylimits_{(a,b]\times(c,d]}|f_{jk}(u_1^-,u_2^-)|\Lambda_{jikl}(\d u_1,\d u_2)\\
&\quad+\;\int\displaylimits_{(a,b]\times(c,d]}|f_{il}(u_1^-,u_2^-)|\Lambda_{ijlk}(\d u_1,\d u_2)\\
&\quad+\;\int\displaylimits_{(a,b]\times(c,d]}|f_{ik}(u_1^-,u_2^-)|\Lambda_{ijkl}(\d u_1, \d u_2)\bigg).
\end{align*}
Summation over $i,k$ and a reordering of some of the resulting sums lead to the inequality
\begin{align*}
\sum_{i,k} \int\displaylimits_{(a,b]\times(c,d]}|O(f)_{ik}|(\d u_1,\d u_2)
&\leq 4\sum_{\substack{i,j,k,l\\i \neq j, k \neq l} } \;\int\displaylimits_{(a,b]\times (c,d]}|f_{ik}(u_1^-,u_2^-)| \Lambda_{ijkl}(\d u_1,\d u_2)\\
&\leq 4 \sum_{i,k}\;\int\displaylimits_{(a,b]\times (c,d]}|f_{ik}(u_1^-,u_2^-)|\, \nu(\d u_1,\d u_2).
\end{align*}
Suppose now that $f(s,s)$ is zero. Then it holds that
\begin{align*}
\sum_{i,k} \int\displaylimits_{(a,b]\times(c,d]}|O(f)_{ik}|(\d u_1,\d u_2)
&\leq 4 \sum_{k,i}\;\int\displaylimits_{(a,b]\times (c,d]}\;\int\displaylimits_{[s,u_1)\times[s,u_2)}|f_{ik}|(\d r_1,\d r_2)\,\nu(\d t_1,\d t_2).
\end{align*}
As a consequence, the norm $\| \cdot \|_{K}$ of $O(f)$ has an upper bound of
 \begin{align} \label{Kontraktion zu vereinfachende Gleichung} \begin{split}
\Vert O(f)\Vert_{K}&=\sum_{i,k}	\;\int\displaylimits_{[s,T]\times[s,T]}e^{-K(\nu_1(u_1)+\nu_2(u_2))}|O(f)_{ik}|(\d u_1,\d u_2)\\
&\leq\sum_{i,k}	\;\int\displaylimits_{[s,T]\times[s,T]}e^{-K(\nu_1(u_1)+\nu_2(u_2))}\;\int\displaylimits_{[s,u_1)\times[s,u_2)}|f_{ik}|(\d r_1,\d r_2)\,\nu(\d u_1,\d u_2)\\
&=\sum_{i,k}	\;\int\displaylimits_{[s,T]\times[s,T]}\;\int\displaylimits_{(r_1,T]\times(r_2,T]}e^{-K(\nu_1(u_1)+\nu_2(u_2))}\nu(\d u_1,\d u_2))|f_{ik}|(\d r_1,\d r_2),
\end{split}\end{align}
where the last equality uses Fubini's theorem.   Moreover, for  arbitrary but fixed $u_1,u_2 \in [s,T]$ let
\begin{align*}
\tilde{\nu}(r_1,r_2)&:=\frac{\mu((u_1,r_1]\times(u_2,r_2])}{\mu((u_1,T]\times(u_2,T])}, \quad  r_1 \in (u_1,T], \, r_2 \in (u_2,T],
\end{align*}
for the same $\mu$ is in equation \eqref{connection beetween mu and nu}.
Without loss of generality we assume that $\mu((u_1,T]\times(u_2,T]) >0$. Otherwise the conclusion that we want to draw is trivial.  Then
\begin{align*}
\tilde\nu_1(r_1)&:=\tilde\nu(r_1,T),\\
\tilde\nu_2(r_2)&:=\tilde\nu(T,r_2),\quad (r_1,r_2)\in(u_1,T]\times(u_2,T],
\end{align*}
correspond to cumulative distribution functions. Let $C:=\tilde{\nu}(u_1,u_2)>0$  and let $(A,B)$ be a random vector that has $\tilde{\nu}$ as its two-dimensional cumulative distribution function. Then, by applying  Sklar's theorem, see e.g.~Theorem 2.3.3 in \cite{nelsen2007introduction}, we can show that
\begin{align}\label{Anwendung Copula theorem}\begin{split}
&\int\displaylimits_{(u_1,T]\times(u_2,T]}e^{-K(\nu_1(r_1)+\nu_2(r_2))}\nu(\d r_1,\d r_2)\\
&\quad \leq Ce^{-K(\nu_1(u_1)+\nu_2(u_2)}\;\int\displaylimits_{(u_1,T]\times (u_2,T]}e^{-CK(\tilde\nu_1(r_1)+\tilde\nu_2(r_2))}\tilde\nu(\d r_1,\d r_2)\\
&\quad =Ce^{-K(\nu_1(u_1)+\nu_2(u_2)}\,\mathbb{E}\left[e^{-CK(\tilde\nu_1(\tilde\nu_1^{-1}(U))+\tilde\nu_2(\tilde\nu_2^{-1}(V)))}\right]
\end{split}\end{align}
for a suitable random vector  $(U,V)$ whose components are uniformly distributed on $(0,1)$ and such that $(A,B)$ and $(\tilde\nu_1^{-1}(U),\tilde\nu_2^{-1}(V))$ have the same distribution.  Note here that the copula of $(U,V)$ may be non-trivial.
The inverse functions $\tilde\nu_1^{-1}$ and $\tilde\nu_2^{-1}$ are here defined as
$ \tilde\nu_n^{-1}(t):=\inf\{x:\tilde\nu_n(x)\geq t\}$, $n=1,2$ for $t\in (0,1)$.
Since $\tilde\nu_1(\tilde\nu_1^{-1}(t)) \geq t$ for $t \in (0,1)$, see e.g.~Theorem 3.1 in \cite{shorack2000probability},  the inequality \eqref{Anwendung Copula theorem} has an upper bound of
\begin{align*}
&\int\displaylimits_{(u_1,T]\times(u_2,T]}e^{-K(\nu_1(r_1)+\nu_2(r_2))}\nu(\d r_1,\d r_2) \leq Ce^{-K(\nu_1(u_1)+\nu_2(u_2)}\,\mathbb{E}\left[e^{-CK(U+V)}\right].
\end{align*}
By applying Theorem 10.6.4 in \cite{kaas2002modern}, we can show that the latter expectation has an upper bound of
\begin{align*}
\mathbb{E}\left[e^{-CK(U+V)}\right]\leq \mathbb{E}\left[e^{-CK(2U)}\right]
\end{align*}
since $x \mapsto \exp\{ - C Kx\}$ is a convex function. As $U$ is uniformly distributed on $(0,1)$, we moreover have
\begin{align*}
\mathbb{E}\left[e^{-CK(2U)}\right] &= \int\displaylimits_{(0,1)}e^{-CK(2u)}\d u\\
& = \frac{1}{2KC}\left(1-e^{-2KC}\right)\\
&\leq  \frac{1}{2KC}.
\end{align*}
We set $K=1$.  Then, all in all we can conclude that  the inequality \eqref{Kontraktion zu vereinfachende Gleichung} has an upper bound of
\begin{align*}
\Vert O(f)\Vert_{K} &\leq  \frac{1}{2K} \sum_{i,k}	\;\int\displaylimits_{[s,T]\times[s,T]}e^{-K(\nu_1(u_1)+\nu_2(u_2))}|f_{ik}|(\d u_1,\d u_2)\\
 &=  \frac{1}{2} \Vert f \Vert_{K}
\end{align*}
whenever $f(s,s)$ equals zero. Suppose now that we have two solutions $R=(P_{ik})_{ik}$ and $R'=(P'_{ik})_{ik}$ of \eqref{2dimKolmForwEq} for given $(P_{i})_{i}$. Then $R-R' \in BV_2^{|\mathcal{Z}|}$ is a fixed point of the operator $O$ and $(R-R')(s,s)$ is zero. So we  have
\begin{align*}
  \Vert R-R' \Vert_{K} \leq \frac{1}{2} \Vert R-R' \Vert_{K},
\end{align*}
which necessarily implies that $R-R'$ is zero on $[s,T]\times [s,T]$.  In an analogous way it is possible to proof $R=R'$  also on  the three rectangles $[0,s]^2$, $[0,s]\times (s,T]$ and $(s,T]\times[0,s]$. Since $T \in (s,\infty)$ was arbitrary, we can expand the uniqueness property to infinity.

With the same ideas that we used in this proof, one can also show the uniqueness of a solution $(P_{i})_{i\in\mathcal{Z}}$ of \eqref{eq: GenKolmForwardEquation} with respect to  $(\Lambda_{ij})_{i,j\in\mathcal{Z}}$. The one-dimensional case is actually even simpler.  The equation system formed by  equations \eqref{eq: GenKolmForwardEquation} and \eqref{2dimKolmForwEq} can then have only one solution  $(P_{i})_{i}$, $(P_{ik})_{ik}$. This completes the proof.
\end{proof}
\section{Conditional  expectations of canonical representations}\label{Section:ExpectationFormulas}

According to \cite{christiansen2021calculation}, for any stochastic process $A$ that has a one-dimensional canonical representation of the form \eqref{def:DefOfB},  it holds that
\begin{align}\label{ExplicFormulasV+} \begin{split}
		\mathbb{E}\bigg[\;\int\displaylimits_{(s,T]} A(\d t) \bigg|\mathcal{G}_s\bigg] &= \sum_{i\in\mathcal{Z}} \;\int\displaylimits_{(s,T]} P_i(t^-)  A_i(\d t)  + \sum_{i,j \in \mathcal{Z}\atop i \neq j} \;\int\displaylimits_{(s,T]}  a_{ij}(t)  P_i(t^-)   \Lambda_{ij}(\d t)
\end{split}\end{align} and \begin{align}\label{ExplicFormulasV-} \begin{split}
\mathbb{E}\bigg[\;\int\displaylimits_{[0,s]}  A(\d t) \bigg|\mathcal{G}_s\bigg] &= \sum_{i\in\mathcal{Z}} \;\int\displaylimits_{[0,s]}  P_i(t^-)  A_i(\d t)+ \sum_{i,j \in \mathcal{Z}\atop i \neq j} \int\displaylimits_{[0,s]}   a_{ij}(t)  P_j(t)   \Lambda_{ij}(\d t)
\end{split}\end{align}
almost surely. These formulas are classical in the life insurance literature in case that $Z$ is a Markov process, and the recent contribution of \cite{christiansen2021calculation} was to show that Markov assumptions are actually not needed here.
 The formulas  \eqref{ExplicFormulasV+} and \eqref{ExplicFormulasV-} are used in life insurance for the calculation of  prospective and retrospective reserves at time $s$.  They are typically applied as follows: For given one-dimensional transition rates,  calculate the corresponding one-dimensional state occupation probabilities by solving \eqref{eq: GenKolmForwardEquation}, and then solve the integrals in \eqref{ExplicFormulasV+} and \eqref{ExplicFormulasV-} in order to obtain the desired conditional expectations. The solving usually happens numerically.

 The following three propositions will allow us to generalize \eqref{ExplicFormulasV+} and \eqref{ExplicFormulasV-} to stochastic processes $A$ that have a two-dimensional canonical representation according to \eqref{def:DefOfB2}.

\begin{proposition}\label{theorem:ProspReserve1}
	Let $A_{ij}$, $i,j\in\mathcal{Z}$,  be a real-valued function on $[0,\infty)^2$  that generates a finite signed measure. Then we  almost surely have
	\begin{align*}
&\mathbb{E}\bigg[\int\displaylimits_{[0,T]^2} I_i(u_1^-)I_j(u_2^-)A_{ij}(\d u_1,\d u_2)\bigg|\mathcal{G}_s\bigg]=\int\displaylimits_{[0,T]^2} P_{ij}(u_1^-,u_2^-)A_{ij}(\d u_1,\d u_2).
	\end{align*}
\end{proposition}
\begin{proof}
	The assertion follows from Fubini's theorem.
\end{proof}
\begin{proposition}\label{theorem:ProspReserve2}
	Let $A_{i}$, $i\in\mathcal{Z}$, be a real-valued function on $[0,\infty)$ that generates a  finite signed measure, and let $a_{ikl}$, $i,k,l\in\mathcal{Z}$, $k\neq l$,  be a measurable and bounded real-valued function on $[0,\infty)^2$.  Then we almost surely have
\begin{align*}
		&\mathbb{E}\bigg[\int\displaylimits_{[0,s]\times [0,T]} I_i(u_1^-)a_{ikl}(u_1,u_2)A_{i}(\d u_1)N_{kl}(\d u_2)\bigg|\mathcal{G}_s\bigg]\\
		&=I_i(s) \int\displaylimits_{[0,s]\times [0,T]} a_{ikl}(u_1,u_2) A_i(\d u_1) \tilde{P}_{kl}(u_2^\pm) \Lambda_{kl}(\d u_2)\\
		&\quad+\sum_{j\in\mathcal{Z}}\;
\int\displaylimits_{[0,s]} \;\int\displaylimits_{[0,T] \times [u_1,s]}  {a}_{ikl}(u_1,u_2)
\tilde{P}_{klij}(u_2^\pm,u_3) \Lambda_{klij}(\d u_2,\d u_3) A_i(\d u_1)
\end{align*}
and
\begin{align*}
		&\mathbb{E}\bigg[\int\displaylimits_{(s,T]\times [0,T]}I_i(u_1^-)a_{ikl}(u_1,u_2)A_{i}(\d u_1)N_{kl}(\d u_2)\bigg|\mathcal{G}_s\bigg]\\
		&=I_i(s) \int\displaylimits_{(s,T]\times [0,T]} a_{ikl}(u_1,u_2) A_i(\d u_1) \tilde{P}_{kl}(u_2^\pm) \Lambda_{kl}(\d u_2)\\
		&\quad+\sum_{j\in\mathcal{Z}}
\int\displaylimits_{(s,T]} \;\int\displaylimits_{[0,T] \times (s,u_1)}   {a}_{ikl}(u_1,u_2)
\tilde{P}_{klji}(u_2^\pm,u_3^-) \Lambda_{klji}(\d u_2,\d u_3) A_i(\d u_1).
\end{align*}
\end{proposition}
\begin{proof} From \eqref{I_represent_by_N}, \eqref{RelationBetweenIandN} we get \begin{align*}
		I_i(t^-)=I_i(s)+\int\displaylimits_{W(t)}\sum_{j\in\mathcal{Z}}N_{ji}^{W(t)}(d u)
	\end{align*} for $W(t):=(s,t)$ in case of  $t> s$ and $W(t):=[t,s]$ in case of  $t\leq s$ and $N_{ij}^{[t,s]}(t):=N_{ji}(t)$ and $N_{ij}^{(s,t)}(t):=N_{ij}(t)$.
  This equation and assumption \eqref{mathcalGs} imply for $D \in \{[0,s], (s,T]\}$ the equation
\begin{align*}
&\mathbb{E}\bigg[\int\displaylimits_{D\times [0,T]}I_i(u_1^-)a_{ikl}(u_1,u_2)A_{i}(\d u_1)N_{kl}(\d u_2)\bigg|\mathcal{G}_s\bigg]\\
&=\mathbb{E}\bigg[\int\displaylimits_{D\times [0,T]}{I}_{i}(s) a_{ikl}(u_1,u_2)A_i(\d u_1)N_{kl}(\d u_2)\bigg|\mathcal{G}_s\bigg]\\
&\quad+\sum_{j\in\mathcal{Z}} \mathbb{E}\bigg[\int\displaylimits_{D}\;\int\displaylimits_{[0,T]\times W(u_1)} \; {a}_{ikl}(u_1,u_2)N_{kl}(\d u_2)N_{ji}^{W(u_1)}(\d u_3) A_i(\d u_1)\bigg|\mathcal{G}_s\bigg].
\end{align*} 
Now apply Fubini's theorem and Campbell's theorem  in order to obtain
\begin{align*}
&\mathbb{E}\Big[\int\displaylimits_{D\times [0,T]}I_i(u_1^-)a_{ikl}(u_1,u_2)A_{i}(\d u_1)N_{kl}(\d u_2)\Big|\mathcal{G}_s\Big]\\
&=\int\displaylimits_{D\times [0,T]}{I}_{i}(s) a_{ikl}(u_1,u_2)A_i(\d u_1)Q_{kl}(\d u_2)\\
&\quad+\int\displaylimits_{D} \int\displaylimits_{[0,T]\times W(u_1)} {a}_{ikl}(u_1,u_2) Q_{klji}^{W(u_1)}(\d u_2, \d u_3) A_i(\d u_1),
\end{align*} for $Q_{klji}^{[0,s]}(t_1,t_2)=Q_{klij}(t_1,t_2)$ and $Q_{klji}^{(s,T]}(t_1,t_2)=Q_{klji}(t_1,t_2)$. Finally, use the equations
\begin{align*}
  Q_{kl}(\d u_2) &= \tilde{P}_{kl}(u_2^\pm) \Lambda_{kl}(\d u_2),\\
  Q_{klij}(\d u_2, \d u_3) &= \tilde{P}_{klij}(u_2^\pm,u_3^\pm) \Lambda_{klij}(\d u_2,\d u_3),
\end{align*}
which follow from the definitions \eqref{DefEqLambda} and \eqref{DGL of g and Gamma}.
\end{proof}
\begin{proposition}\label{theorem:ProspReserve3}
		Let $a_{ijkl}$, $i,j,k,l\in\mathcal{Z}$, $i\neq j$, $k\neq l$, be a measurable and bounded real-valued function on $[0,\infty)^2$.
Then we almost surely have
 	\begin{align*}	&\mathbb{E}\bigg[\;\int\displaylimits_{[0,T]^2}a_{ijkl}(u_1,u_2)N_{ij}(\d u_1)N_{kl}(\d u_2)\bigg|\mathcal{G}_s\bigg]\\
		&=\;\int\displaylimits_{[0,T]^2}a_{ijkl}(u_1,u_2)\tilde P_{ijkl}(u_1^\pm,u_2^\pm)\Lambda_{ijkl}(\d u_1,\d u_2).
	\end{align*}
\end{proposition}
\begin{proof}
	The result follows directly from Campbell's theorem and \eqref{DGL of g and Gamma}.
\end{proof}
\begin{corollary}\label{corollary:ProspReserve}
	Suppose that $A$ is a stochastic process that has  a two-dimensional  canonical representation  according to \eqref{def:DefOfB2}. Then  we almost surely have
\begin{align}\begin{split} \label{formula:ProspReserve}
			&\mathbb{E}\bigg[\;\int\displaylimits_{[0,T]^2}  A(\d u_1,\d u_2) \bigg|\mathcal{G}_s\bigg]\\
			&=\sum_{i,j\in\mathcal{Z}} \;\int\displaylimits_{[0,T]^2}P_{ij}(u_1^-,u_2^-)A_{ij}(\d u_1,\d u_2)\\
			&\quad+\sum_{\substack{i,k,l\\k\neq l}}I_i(s) \int\displaylimits_{[0,T]^2}a_{ikl}(u_1,u_2) A_i(\d u_1) \tilde{P}_{kl}(u_2^\pm) \Lambda_{kl}(\d u_2)\\
			&\quad+\sum_{\substack{i,j,k,l\\k\neq l}}\;\;\int\displaylimits_{[0,s]}\; \int\displaylimits_{[0,T] \times [u_1,s]}  {a}_{ikl}(u_1,u_2)
\tilde{P}_{klij}(u_2^\pm,u_3) \Lambda_{klij}(\d u_2,\d u_3) A_i(\d u_1)\\
			&\quad+\sum_{\substack{i,j,k,l\\k\neq l}}\;\;\int\displaylimits_{(s,T]}\; \int\displaylimits_{[0,T] \times (s,u_1)}  {a}_{ikl}(u_1,u_2)
\tilde{P}_{klji}(u_2^\pm,u_3^-) \Lambda_{klji}(\d u_2,\d u_3) A_i(\d u_1)\\
			&\quad+\sum_{\substack{i,j,k,l\\i\neq j,k\neq l}}\;\int\displaylimits_{[0,T]^2}a_{ijkl}(u_1,u_2)\tilde P_{ijkl}(u_1^\pm,u_2^\pm)\Lambda_{ijkl}(\d u_1,\d u_2).
		\end{split}
\end{align}
\end{corollary}
In insurance practice,  this formula may be used as follows:
For given one-dimensional and two-dimensional  transition rates,  calculate the corresponding one-dimensional and two-dimensional state occupation probabilities by solving \eqref{eq: GenKolmForwardEquation} and \eqref{2dimKolmForwEq}. Then solve the integrals in \eqref{formula:ProspReserve}  in order to obtain the  desired conditional expectation.

\section{Example: Conditional variance of the future liabilities}\label{Chap:5 prospective reserve}

We still think of $s$ as an arbitrary but fixed parameter and omit this parameter in the notation.
The discounted accumulated future payments $Y^+$  according to definition \eqref{def:DefOfY} are usually not measurable with respect to $\mathcal{G}_s$, so actuaries use the so-called \emph{prospective reserve} at time $s$
\begin{align}\label{V+}
  V^+&:= \E[ Y^+ | \mathcal{G}_s ]
 \end{align}
 as a proxy for $Y^+$. 
In order to quantify the dispersion of the approximation error $Y^+-V^+$, the actuary is also interested in the  conditional variance
\begin{align}\label{Var+}
 \textrm{Var}[ Y^{+} | \mathcal{G}_s ]:&= \E[ (Y^{+})^2 | \mathcal{G}_s ] - (\E[ Y^{+} | \mathcal{G}_s ])^2.
 \end{align}
This section discusses the calculation of $V^+$ and
\begin{align}\label{S+}
  S^+:= \E[ (Y^+)^2 | \mathcal{G}_s ],
 \end{align}
 from which we can then get the conditional variance as
\begin{align*}
   \textrm{Var}[ Y^{+} | \mathcal{G}_s ]=S^+ - (V^+)^2.
\end{align*}
Our results here are limited to insurance cash-flows $B$ that have a one-dimensional canonical representation,
\begin{align}
 B(t)=\sum_{i} \;\int\displaylimits_{[0,t]} I_i(u^-) B_i(\d u)+\sum_{i,j:j\neq i}\;\int\displaylimits_{[0,t]} b_{ij}(u) N_{ij}(\d u),\quad t\geq 0.
\end{align}
From \eqref{def:DefOfY} and  \eqref{ExplicFormulasV+} we almost surely obtain
\begin{align*}
  V^+ &=  \sum_{i\in\mathcal{Z}} \;\int\displaylimits_{(s,T]} \frac{\kappa(s)}{\kappa(t)} P_i(t^-)  B_i(\d t)   + \sum_{i,j \in \mathcal{Z}\atop i \neq j} \;\int\displaylimits_{(s,T]}  \frac{\kappa(s)}{\kappa(t)} b_{ij}(t)  P_i(t^-)   \Lambda_{ij}(\d t).
\end{align*}
This formula and equation \eqref{eq: GenKolmForwardEquation} allow us to calculate $V^+$ from given one-dimensional transition rates $(\Lambda_{ij}(t))_{t > s}$, $i,j \in \mathcal{Z}$.

We now turn to the calculation of $S^+$. Analogously to \eqref{def:RepOfBsquared}, one can show that
\begin{align*}
	(Y^+)^2 &= \sum_{i,j\in\mathcal{Z}}\int\displaylimits_{(s,T]^2}\frac{2 \kappa(s)^2}{\kappa(u_1)\kappa(u_2)} I_i(u_1^-)I_j(u_2^-) B_i(\d u_1) B_j(\d u_2)\\
		&\quad+\sum_{\substack{i,j,k\in\mathcal{Z}\\j\neq k}}\;\int\displaylimits_{(s,T]^2}\frac{2\kappa(s)^2}{\kappa(u_1)\kappa(u_2)} I_i(u_1^-)  b_{jk}(u_2) B_i(\d u_1)N_{jk}(\d u_2) \\
		&\quad+\sum_{\substack{i,j,k,l\in\mathcal{Z}\\i\neq j,k\neq l}}\int\displaylimits_{(s,T]^2}\frac{2\kappa(s)^2}{\kappa(u_1)\kappa(u_2)} b_{ij}(u_1) b_{kl}(u_2) N_{ij}(\d u_1) N_{kl}(\d u_2).
\end{align*}
By using definition \eqref{S+}, interchanging $\int_{(s,T]^2} $ and $\int_{[0,T]^2}  \mathds{1}_{(s,T]^2} $, and  applying Corollary \ref{corollary:ProspReserve}, we obtain that
\begin{align*}
S^+&=\sum_{\substack{i,j\in\mathcal{Z}}}\;\int\displaylimits_{(s,T]^2}\frac{2\kappa(s)^2}{\kappa(u_1)\kappa(u_2)} P_{ij}(u_1^-,u_2^-)B_i(\d u_1) B_j(\d u_2)\\
&\quad+\sum_{\substack{i,k,l\in\mathcal{Z}\\k\neq l}}I_i(s)\;\int\displaylimits_{(s,T]^2}\frac{2\kappa(s)^2}{\kappa(u_1)\kappa(u_2)} b_{kl}(u_2)P_{k}(u_2^-) B_i(\d u_1)\Lambda_{kl}(\d u_2)\\
&\quad +\sum_{\substack{i,j,k,l\in\mathcal{Z}\\ k\neq l}}\;\int\displaylimits_{(s,T]}\;\int\displaylimits_{(s,T]\times (s,u_1)}\frac{2\kappa(s)^2}{\kappa(u_1)\kappa(u_2)} b_{kl}(u_2)  P_{kj}(u_2^-,u_3^-)\Lambda_{klji}(\d u_2,\d u_3)  B_i(\d u_1)\\
&\quad+\sum_{\substack{i,j,k,l\in\mathcal{Z}\\i\neq j,k\neq l}}\quad\;\int\displaylimits_{(s,T]^2} \frac{2\kappa(s)^2}{\kappa(u_1)\kappa(u_2)} b_{ij}(u_1) b_{kl}(u_2)P_{ik}(u_1^-,u_2^-)\Lambda_{ijkl}(\d u_1,\d u_2)
\end{align*}
almost surely.
 This formula and the equations \eqref{eq: GenKolmForwardEquation} and \eqref{2dimKolmForwEq} allow us to calculate $S^+$ from given  transition rates $(\Lambda_{ij}(t))_{t \in (s,T] }$ and  $(\Lambda_{ijkl}(t))_{t\in (s,T]^2}$, $i,j,k,l \in \mathcal{Z}$.

For the conditional expectation and the conditional variance of $Y^-$ we can obtain a similar result. We leave these analogous calculations to the reader.

\section{Example: Prospective and retrospective reserve for a  path-dependent cash-flow}\label{Chap:7 retrospective calculation}

This section continues with Example \ref{ExampleFPOCashFlow}.  Recall that  \eqref{V+} is the so-called prospective reserve at time $s$. The \emph{retrospective reserve} at time $s$ is defined as
\begin{align}\label{V-}
  V^-&:= \E[ Y^- | \mathcal{G}_s ]
 \end{align}
for $Y^-$ defined by \eqref{def:DefOfY-}. Analogously to \eqref{eq: free-policy-option-cash-flow} one can show that
\begin{align*}
  Y^-&= \sum_{i\in \mathcal{S}_0}\;\int\displaylimits_{[0,s]} \frac{\kappa(s)}{\kappa(u)} I_i(u^-) C_i(\d u)+\sum_{\substack{k,l\in \mathcal{S}_0\\k\neq l}}\;\int\displaylimits_{[0,s]} \frac{\kappa(s)}{\kappa(u)}  c_{kl}(u)  N_{kl}(\d u)\\
	 &\quad +  \sum_{k \in \mathcal{S}_0}\sum_{i,l \in \mathcal{S}_1} \;\int\displaylimits_{[0,s]^2} \frac{\kappa(s)}{\kappa(u_1)}  I_i(u_1^-)\rho(u_2,k,l) C_{i}(\d u_1)   N_{kl}(\d u_2)\\
	 &\quad +\sum_{k \in \mathcal{S}_0}\sum_{\substack{i,j,l\in \mathcal{S}_1\\i \neq j}}\;\int\displaylimits_{[0,s]^2}  \frac{\kappa(s)}{\kappa(u_1)}  \rho(u_2,k,l) c_{ij}(u_1)    N_{ij}(\d u_1)  N_{kl}(\d u_2).
\end{align*}
By using \eqref{ExplicFormulasV-}, interchanging $\int_{[0,s]} $ and $\int_{[0,T]}  \mathds{1}_{[0,s]} $, interchanging $\int_{[0,s]^2}$ and $\int_{[0,T]^2}  \mathds{1}_{[0,s]^2} $, and applying Corollary \ref{corollary:ProspReserve}, we can conclude that
\begin{align*}
  V^-&= \sum_{i\in\mathcal{S}_0} \;\int\displaylimits_{[0,s]} \frac{\kappa(s)}{\kappa(u)} P_i(u^-)  C_i(\d u)   + \sum_{\substack{k,l\in\mathcal{S}_0\\ k \neq l}} \;\int\displaylimits_{[0,s]}  \frac{\kappa(s)}{\kappa(u)} c_{kl}(u)  P_l(u)   \Lambda_{kl}(\d u)\\
  &\quad+\sum_{k\in\mathcal{S}_0}\sum_{l,i\in\mathcal{S}_1}I_i(s)\int\displaylimits_{[0,s]^2}\frac{\kappa(s)}{\kappa(u_1)}\rho(u_2,k,l) P_{l}(u_2)C_i(\d u_1)\Lambda_{kl}(\d u_2)\\
  &\quad +\sum_{k\in\mathcal{S}_0}\sum_{l,i\in\mathcal{S}_1}\sum_{j\in\mathcal{Z}}\;\int\displaylimits_{[0,s]}\;\int\displaylimits_{[0,s]\times [u_1,s]}\frac{\kappa(s)}{\kappa(u_1)}\rho(u_2,k,l)P_{lj}(u_2,u_3)\Lambda_{klij}(\d u_2,\d u_3) C_i(\d u_1)\\
  & \quad + \sum_{k\in\mathcal{S}_0}\sum_{\substack{l,i,j\in\mathcal{S}_1\\i\neq j}}\;\int\displaylimits_{[0,s]^2}\frac{\kappa(s)}{\kappa(u_1)}\rho(u_2,k,l) c_{ij}(u_1)P_{jl}(u_1,u_2)\Lambda_{ijkl}(\d u_1,\d u_2)
\end{align*}
almost surely. This formula and the equations \eqref{eq: GenKolmForwardEquation} and \eqref{2dimKolmForwEq} allow us to calculate $V^-$ from given  transition rates $(\Lambda_{ij}(t))_{t \leq s}$ and  $(\Lambda_{ijkl}(t))_{t \leq s}$, $i,j,k,l \in \mathcal{Z}$.

By arguing analogously to  \eqref{eq: free-policy-option-cash-flow}, one can moreover show that the discounted future  liabilities $Y^+$ have the representation
\begin{align*}
  Y^+&= \sum_{i\in \mathcal{S}_0}\;\int\displaylimits_{(s,T]} \frac{\kappa(s)}{\kappa(u)} I_i(u^-) C_i(\d u)+\sum_{\substack{k,l\in \mathcal{S}_0\\k\neq l}}\;\int\displaylimits_{(s,T]} \frac{\kappa(s)}{\kappa(u)}  c_{kl}(u)  N_{kl}(\d u)\\
	 &\quad +  \sum_{k \in \mathcal{S}_0}\sum_{i,l \in \mathcal{S}_1} \;\int\displaylimits_{(s,T]\times [0,T]} \frac{\kappa(s)}{\kappa(u_1)}  I_i(u_1^-)\rho(u_2,k,l) C_{i}(\d u_1)   N_{kl}(\d u_2)\\
	 &\quad +\sum_{k \in \mathcal{S}_0}\sum_{\substack{i,j,l\in \mathcal{S}_1\\i \neq j}}\;\int\displaylimits_{(s,T]\times [0,T]}  \frac{\kappa(s)}{\kappa(u_1)}  \rho(u_2,k,l) c_{ij}(u_1)    N_{ij}(\d u_1)  N_{kl}(\d u_2).
\end{align*}
By applying \eqref{ExplicFormulasV+}, interchanging $\int_{(s,T]} $  and $\int_{[0,T]}  \mathds{1}_{(s,T]} $, interchanging $\int_{(s,T]\times [0,T]} $ and $\int_{[0,T]^2}  \mathds{1}_{(s,T] \times [0,T]}  $, and applying Corollary \ref{corollary:ProspReserve}, we almost surely get
\begin{align*}
				V^+&=\sum_{i\in\mathcal{S}_0} \;\int\displaylimits_{(s,T]} \frac{\kappa(s)}{\kappa(u)} P_i(u^-)  C_i(\d u)  + \sum_{\substack{i,j\in\mathcal{S}_0\\ i \neq j}} \;\int\displaylimits_{(s,T]}  \frac{\kappa(s)}{\kappa(u)} c_{ij}(u)  P_i(u^-)   \Lambda_{ij}(\d u)\\ &\quad+\sum_{k\in\mathcal{S}_0}\sum_{l,i\in\mathcal{S}_1}I_i(s) \int\displaylimits_{(s,T]\times[0,T]} \frac{\kappa(s)}{\kappa(u_1)}\rho(u_2,k,l)\tilde P_{kl}(u_2^\pm)C_i(\d u_1)\Lambda_{kl}(\d u_2)\\ &\quad+\sum_{k\in\mathcal{S}_0}\sum_{l,i\in\mathcal{S}_1}\sum_{j\in\mathcal{Z}}\int\displaylimits_{(s,T]}\;\int\displaylimits_{[0,T]\times (s,u_1)}\frac{\kappa(s)}{\kappa(u_1)}\rho(u_2,l,k)\tilde P_{klji}(u_2^{\pm},u_3^-)\Lambda_{klji}(\d u_2,\d u_3)C_i(\d u_1)\\
&\quad+\sum_{k\in\mathcal{S}_0}\sum_{\substack{l,i,j\in\mathcal{S}_1\\i\neq j}} \int\displaylimits_{(s,T]\times[0,T]}\frac{\kappa(s)}{\kappa(u_1)}\rho(u_2,k,l)c_{ij}(u_1)\tilde P_{ijkl}(u_1^-,u_2^\pm)\Lambda_{ijkl}(\d u_1,\d u_2).
\end{align*}
This formula and the equations \eqref{eq: GenKolmForwardEquation} and \eqref{2dimKolmForwEq} allow us to calculate $V^-$ from given  transition rates $(\Lambda_{ij}(t))_{t \leq T}$ and  $(\Lambda_{ijkl}(t))_{t \leq T}$, $i,j,k,l \in \mathcal{Z}$.

\section{Conclusion and outlook}\label{SectionConclusion}

So far, the non-Markov calculation technique of Christiansen \cite{christiansen2021calculation} has been limited to the calculation of first-order moments.  By introducing two-dimensional forward and backward transition rates, we make second-order moments accessible as well. The two-dimensional  rates capture intertemporal dependency structures that the one-dimensional rates miss. In order to calculate also third-order or even higher-order moments, one may envision further extensions of the forward and backward transition rate concept to three dimensions  or even higher dimensions. Such extensions are beyond the scope of this paper and left for future research.

Intertemporal dependency structures play a major role when life insurance cash-flows are path-dependent. Such path dependencies occur for example upon contract modifications.
We illustrated how two-dimensional forward and backward transition rates are of help here for actuarial calculations by studying insurance policies with free-policy option.

In insurance practice the one-dimensional and two-dimensional forward and backward transition rates have to be estimated from data. For the one-dimensional case, the landmark Nelson-Aalen estimator of Putter and Spitoni \cite{putter2018non} can be suitably adapted to do that task, see Christiansen \cite{christiansen2021calculation}. For the two-dimensional case, efficient estimators still have to be explored.

\appendix

\section{Appendix}
\begin{theorem}[Campbell theorem]\label{theo:Campbell}
Let $\eta$ be a point process on $(\mathbb{R}^d,\mathbb{B}(\mathbb{R}^d))$ with intensity measure $\lambda$ and let $u:\mathbb{R}^d\rightarrow\mathbb{R}$ be a measurable function. Then \begin{align*}
\;\int\displaylimits u(x)\eta(\d x)
\end{align*} is a random variable and \begin{align*}
\mathbb{E}\left[\;\int\displaylimits u(x)\eta(\d x)\right]=\;\int\displaylimits u(x)\lambda(\d x),
\end{align*}
whenever $u\geq0$ or $\;\int\displaylimits | u(x)| \lambda(\d x)<\infty$.
\end{theorem}
For the proof see section 2.2 in \cite{last2017lectures}.

\begin{proposition}[integration by parts]\label{prop: partial integration}
	Suppose that $(F(t))_{t\geq 0}$ and $(G(t))_{t\geq 0}$ are real-valued c\`{a}dl\`{a}g  processes with paths of finite variation. Then 
	\begin{align*}
		\int\displaylimits_{(a,b]}F(x)\d G(x)=F(b)G(b)&-F(a)G(a)-\int\displaylimits_{(a,b]}G(x^-)\d F(x)
	\end{align*}
	for $(a,b]\subset (0,\infty)$.
\end{proposition}
\begin{proof}
	For the proof see corollary 2 following theorem 22  in \cite{protter2005stochastic}.
\end{proof}
\bibliographystyle{acm}
\bibliography{mybib.bib}

\end{document}